%% file: main.tex
\documentclass[letterpaper,10pt,conference]{ieeeconf}
\IEEEoverridecommandlockouts
\usepackage{graphicx}
\usepackage{subcaption}
\usepackage{etoolbox} 
\usepackage{scalerel} 
\usepackage{siunitx}
\usepackage{pifont}
\usepackage{graphicx}
\usepackage{amsmath,amssymb,textcomp}
\usepackage{algorithm} 
\usepackage{cite}
\usepackage{booktabs}
\usepackage{url}
\usepackage{algpseudocode}
\usepackage{paralist}
\usepackage{color}
\usepackage[T1]{fontenc}
\usepackage{stmaryrd}
\usepackage{epsfig} 
\usepackage{bm} 
\usepackage{listings,lstautogobble}
\usepackage{mathtools}
\usepackage{mathrsfs}
\usepackage{xspace}
\usepackage{verbatim} 
\usepackage{epstopdf}
\usepackage{listings}
\usepackage{parcolumns}
\usepackage{color}
\usepackage{mathtools}
\usepackage{mathrsfs}
\usepackage{xspace}
\usepackage{verbatim}
\usepackage[utf8]{inputenc}
\usepackage{calrsfs}
\usepackage{rotating,multirow}
\usepackage{mathtools}
\usepackage{hyperref}
\usepackage[most]{tcolorbox}

\usepackage[T1]{fontenc}
\usepackage[utf8]{inputenc}

\usepackage{amsthm}
\newtheorem{remark}{Remark}
\newtheorem{lemma}{Lemma}

\newtheorem{theorem}{Theorem}
\newtheorem{definition}{Definition}
\newtheorem{proposition}{Proposition}

\allowdisplaybreaks

\renewcommand{\footnotesize}{\fontsize{8pt}{11pt}\selectfont}

\DeclareMathAlphabet{\mathcal}{OMS}{cmsy}{m}{n}

\DeclarePairedDelimiter{\norm}{\lVert}{\rVert}

\newcommand{\R}{{\mathbb{R}}}

\newcommand{\argmin}{\textrm{arg}\min}

\usepackage[algo2e]{algorithm2e}
\usepackage[]{algpseudocode}

\newcommand{\setdef}[2][]{
	\left\{
	\ifblank{#1}{}{#1 \hspace{.1cm} \middle| \hspace{.1cm}}
	#2 
	\right\}
}
\newcommand{\zono}[1]{\langle #1 \rangle}

\newcommand{\Rpredict}{\tilde{\mathcal{R}}}
\newcommand{\Rmeas}{\hat{\mathcal{R}}}

\newcommand{\Cpredict}{\tilde{\mathcal{C}}}
\newcommand{\Cmeas}{\hat{\mathcal{C}}}

\def\algref#1{Algorithm~\ref{#1}}

\def\secref#1{Sec.~\ref{#1}}
\def\defref#1{Definition~\ref{#1}}
\def\figref#1{Fig.~\ref{#1}}
\def\propref#1{Proposition~\ref{#1}}

\def\eqref#1{(\ref{#1})}
\def\eqnref#1{(\ref{#1})}

\title{\LARGE \bf Data-Driven Set-Based Estimation using Matrix Zonotopes \\ with Set Containment Guarantees
  \thanks{
    $^*$Authors are with equal contributions. $^{1}$The author is with Jacobs University, Bremen. \texttt{a.alanwar@jacobs-university.de}. $^{2}$The authors are with the Division of Decision and Control Systems at KTH Royal Institute of Technology. \texttt{\{alberndt, hsan, kallej\}@kth.se}.
  }
}

\author{Amr~Alanwar$^{*,1}$, Alexander~Berndt$^{*,2}$, Karl~Henrik~Johansson$^{2}$, and Henrik~Sandberg$^{2}$}

\begin{document}

\maketitle
\input{Sections/0-abs.tex}

\input{Sections/1-intro.tex}
\input{Sections/2-preliminaries.tex}
\input{Sections/3-estimation.tex}

\input{Sections/4-evaluation.tex}
\input{Sections/5-conclusions.tex} 

\bibliographystyle{ieeetr}
\bibliography{ref}

\end{document}

%% file: Sections/0-abs.tex
\begin{abstract}

We propose a method to perform set-based state estimation of an unknown dynamical linear system using a data-driven set propagation function. Our method comes with set-containment guarantees, making it applicable to safety-critical systems.
The method consists of two phases: (1) an offline learning phase where we collect noisy input-output data to determine a function to propagate the state-set ahead in time; and (2) an online estimation phase consisting of a time update and a measurement update. It is assumed that known finite sets bound measurement noise and disturbances, but we assume no knowledge of their statistical properties. These sets are described using zonotopes, allowing efficient propagation and intersection operations. We propose a new approach to compute a set of models consistent with the data and noise-bound, given input-output data in the offline phase. The set of models is utilized in replacing the unknown dynamics in the data-driven set propagation function in the online phase. Then, we propose two approaches to perform the measurement update. Simulations show that the proposed estimator yields state sets comparable in volume to the $3\sigma$ confidence bounds obtained by a Kalman filter approach, but with the addition of state set-containment guarantees. We observe that using constrained zonotopes yields smaller sets but with higher computational costs than unconstrained ones.
\end{abstract}

%% file: Sections/1-intro.tex
\section{Introduction}
Set-based estimation 
  involves the computation of
   a set, which is guaranteed to contain the system's
   true state at each time step  
   given bounded uncertainties \cite{conf:setmem1971}. 
Existing set-based observers require a system model 
  to propagate the state set at each time step \cite{conf:reducedsetbased,conf:distributedestimation}. 
We address the problem of propagating the state set 
  using only noisy offline input-output data and 
  merging this with online measurements 
  to obtain a time-varying state set
  which is guaranteed to contain the true system's state 
  at each time-step. 
This problem is essential in
  safety-critical applications \cite{conf:thesisalthoff}.

Two popular set-based estimators are 
interval observers
  and set-membership observers. 
Interval-based observers generally 
  generate state estimates by utilizing 
  an observer gain to fuse 
  a model-based time update of the state with 
  current measurements. 
For example, the authors in \cite{conf:interval4} 
  propose an exponentially stable interval-based observer for 
  time-invariant linear systems.
Set-membership observers generally follow a geometrical approach 
  by intersecting the state-space regions consistent 
  with the model with those from the measurements 
  to obtain the current state set
  \cite{conf:orthotope}.
This approach has been extended to sensor networks  
  with event-based communication in \cite{conf:disevent}
  and multi-rate systems in \cite{conf:intermulti}.
Various set representations have been used for set-membership observers
  such as ellipsoids \cite{conf:ellipsoide}, 
  polytopes \cite{conf:polytope} 
  and zonotopes \cite{conf:set-diff}.
Zonotopes are a special class of polytopes for which 
  one can efficiently compute linear maps, and Minkowski 
  sums -- both frequent operations performed by set-based observers.
   
All the aforementioned observers 
  use a model of the underlying system 
  to propagate the state set.
However, identifying a system model is often time-consuming, and 
  the identified model is not necessarily well-suited 
  for estimation or control.
Recent works based on Willems' fundamental lemma \cite{conf:willems}
  have shown that system trajectories can be 
  used directly to synthesize controllers.
The authors in \cite{Alpago2020EKF_DeepC} 
  present an extended Kalman filter and model predictive control (MPC)
  scheme computed directly from system trajectories.
Stability and robustness guarantees
  for such a data-driven control scheme are presented in \cite{conf:formulas},
  and for an MPC scheme in \cite{conf:mpcguarantees}. 
An alternative approach is to find a set of models that is consistent 
  with data and use this set of models to propagate 
  a state set \cite{conf:datadriven_reach}. 
  


Our contribution is a 
  novel method to perform set-based state estimation 
  with set-containment guarantees 
  given bounded, noisy measurements and known inputs.
The algorithm, summarized in Fig. \ref{fig:method}, consists of an \textit{offline learning phase} 
  to determine a state-propagation function $f(\cdot)$ directly 
  from data, and an \textit{online estimation phase} to perform
  a time update using $f(\cdot)$ and measurements iteratively to
  track the system state. A new approach to compute the set of models consistent with the data and noise bound from input-output data is proposed different from input-state data in \cite{conf:datadriven_reach,conf:ourjournal}. Then, we present two approaches to perform the measurement update utilizing 
  either the singular value decomposition (SVD) of the observation matrix or
  an optimization formulation.
We compare the approaches in simulation.
Our method is shown to yield set-based state estimates 
  similar in size to $3\sigma$ confidence bounds of
  an approach based on system identification and a Kalman filter, 
  but with the addition of set-containment guarantees. The code to recreate our findings is publicly available\footnotemark.

\footnotetext{\href{https://github.com/alexberndt/data-driven-set-based-estimation-zonotopes}{https://github.com/alexberndt/data-driven-set-based-estimation-zonotopes}}

The rest of this paper is outlined as follows. 
\secref{sec:preliminaries} introduces the preliminaries and problem statement.
We present our method in \secref{sec:method} and 
evaluate it in \secref{sec:evaluation}. Finally, 
\secref{sec:conclusions} concludes the paper.

%% file: Sections/2-preliminaries.tex
\begin{figure*}[t] 
  \centering 
  \includegraphics[width=0.92\linewidth]{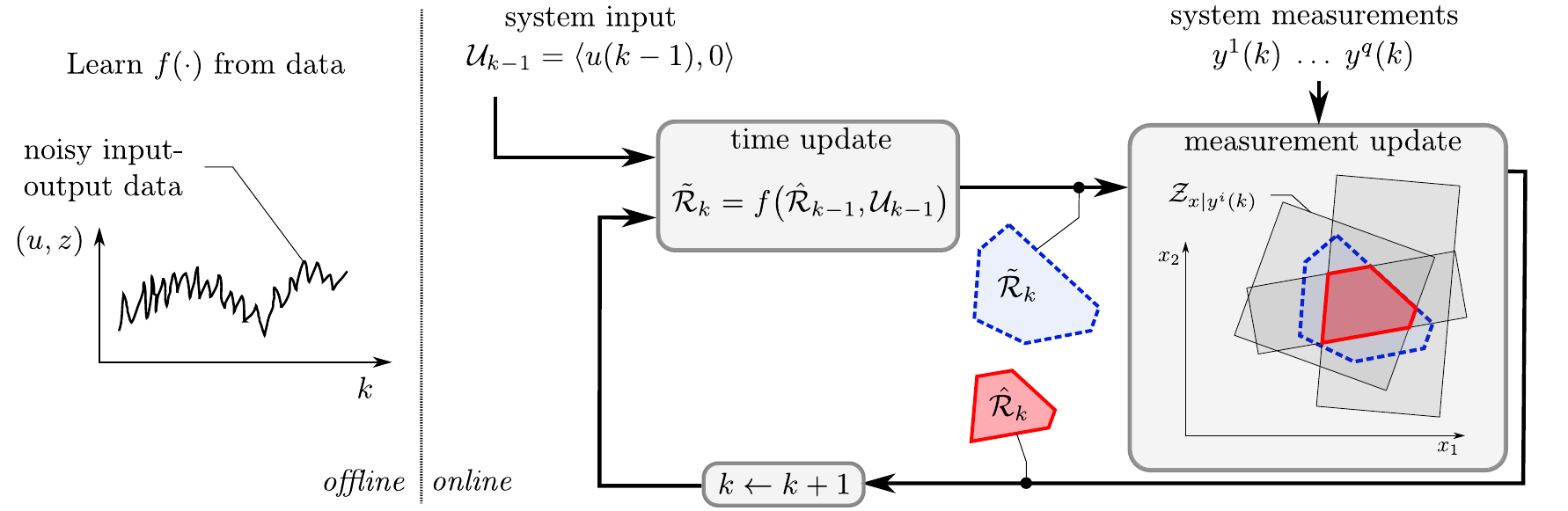}
  \caption{The proposed method showing the offline learning phase 
    yielding $f(\cdot)$ and the online estimation phase
    which utilizes $f(\cdot)$ to perform the time update, followed 
    by a measurement update yielding the set $\hat{\mathcal{R}}_k$ at time-step $k$.
  }
  \label{fig:method}
\end{figure*}

\section{Preliminaries and Problem Statement} 
\label{sec:preliminaries}

We denote the $i$-th element of a vector or list $A$ by $A^{(i)}$. We first introduce some set representations. 
\begin{definition}
{\normalfont(Zonotope \cite{conf:zono1998})}
\label{def:zonotope} 
    Given a center $c \in \mathbb{R}^n$ and a number $\xi \in \mathbb{N}$
      of generator vectors in a generator matrix 
      $G=[g^{(1)},...,g^{(\xi)}] \in \mathbb{R}^{n \times \xi}$, 
      a zonotope is a set
    \begin{equation}
      \mathcal{Z} = \Big\{ x \in \mathbb{R}^n \; \Big| \; x = c + \sum_{i=1}^\xi \beta^{(i)} \, g^{(i)} \, ,
      -1 \leq \beta^{(i)} \leq 1 \Big\}.
    \end{equation} 
    We use the shorthand notation $\mathcal{Z} = \zono{c,G}$.
\end{definition}
Given two zonotopes $\mathcal{Z}_1$ and $\mathcal{Z}_2$, 
  we use the notation $+$  
  for the Minkowski sum, 
  and $\mathcal{Z}_1 - \mathcal{Z}_2$ 
  to denote $\mathcal{Z}_1 + (- \mathcal{Z}_2)$ not the Minkowski difference.



\begin{definition} \label{def:mat_zonotope}
{\normalfont(Matrix zonotope \cite[p.52]{conf:thesisalthoff})}
  Given a center matrix $C \in \mathbb{R}^{n\times k}$ 
    and $\xi \in \mathbb{N}$ generator matrices 
    ${G}^{(i)} \in \mathbb{R}^{n \times k}$ 
    where $i \in \{1,\dots,\xi\}$,  
    a matrix zonotope is the set
  \begin{equation*}
    \mathcal{M} = 
    \Big\{   X \in \mathbb{R}^{n\times k} 
             \; \Big| \; 
             X = C + \sum_{i=1}^\xi {\beta}^{(i)} \, {G}^{(i)} \, 
             ,
             -1 \leq {\beta}^{(i)} \leq 1 
    \Big\}.
  \end{equation*}
  We use the notation $\mathcal{M} = \zono{C,{G}^{(1:\xi)}}$,
  where ${G}^{(1:\xi)}=[{G}^{(1)},\dots,{G}^{(\xi)}]$.
\end{definition}

\begin{definition} \label{def:intmat}(\normalfont{Interval matrix} \cite[p. 42]{conf:thesisalthoff})
An interval matrix $\mathcal{I}$ specifies the interval of all possible values for each matrix element between the left limit $\underline{I}$ and right limit $\bar{I}$:
\begin{align}
    \mathcal{I} = \begin{bmatrix} \underline{I},\bar{I}  \end{bmatrix}, \quad \underline{I},\bar{I} \in \mathbb{R}^{r \times c}
\end{align}
\end{definition}






We consider estimating the 
  set of all possible system states 
  using an array of $q$ sensors. Our system is described as 
  \begin{subequations}\label{eq:sys}
  \begin{align}
          x(k+1) &= A_{\text{tr}}x(k) + B_{\text{tr}} u(k) + w(k), \label{eq:underlying_system} \\
          y^i(k) &= C^i x(k) + v^i(k), \;\; i \in \{1,\dots,q \} \label{eq:observations}, 
  \end{align}
  \end{subequations}
where 
  $x(k) \in \mathbb{R}^n$ is the system state, 
  $u(k) \in \mathbb{R}^m$ the input,
  $y^i(k) \in \mathbb{R}^{p_i}$ the measurement of sensor $i$, 
  $x(0) \in \mathcal{X}_0$ the initial condition
  where $\mathcal{X}_0$ is the initial bounding zonotope.
Furthermore, the system matrices
  $A_{\text{tr}} \in \mathbb{R}^{n \times n}$ and $B_{\text{tr}} \in \mathbb{R}^{n \times m}$
  are \textit{unknown} whereas
  $C^i \in \mathbb{R}^{p_i \times n}$ is \textit{known} 
  for all $i \in \{1,\dots,q \}$. 
The noise $w(k) \in \mathcal{Z}_w$  
  and $v^i(k) \in \mathcal{Z}_{v,i}$ 
  are assumed to belong to the bounding zonotopes $\mathcal{Z}_w  = \zono{c_w,G_w} \subset \mathbb{R}^{n}$ 
  and $\mathcal{Z}_{v,i} = \zono{c_{v,i},G_{v,i}} \subset \mathbb{R}^{p_i}$ for $i \in \{1,\dots,q \}$, respectively. 
We denote the Frobenius norm by $\|.\|_F$ 
  and the null space of a matrix $A$ by $\texttt{ker}(A)$. 
  We compute the pseudoinverse of an interval matrix by adapting \cite[Thm 2.40]{conf:inverseInterval}. The pseudoinverse of an interval matrix is denoted by $\dagger$.
  
Let ${\mathcal{R}}_k$ denote a set containing $x(k)$ 
  given the \textit{exact} system model
  and bounded, but \textit{unknown}, process and measurement noise.
The problem addressed in this paper is to develop an algorithm 
  that returns a set
  $\hat{\mathcal{R}}_k \supseteq {\mathcal{R}}_k$,
  which is \textit{guaranteed} to contain the true state $x(k)$
  at each time instance $k$, i.e., $x(k) \in \hat{\mathcal{R}}_k$ for all $k$,
  given input-output data and bounds for model uncertainties and measurement noise without knowledge of the model $\begin{bmatrix}A_{\text{tr}} & B_{\text{tr}} \end{bmatrix}$.

%% file: Sections/3-estimation.tex
\section{Data-driven Set-based Estimation}
\label{sec:method}

Our proposed data-driven set estimator consists of two phases: 
  an \textit{offline learning phase} and an \textit{online estimation phase}.
In the offline phase, we compute the function 
  to perform the time update. The online phase consists of iteratively performing 
  a time update and a measurement update. We denote the time and measurement updated sets at $k$ by $\Rpredict_k \subset \mathbb{R}^n$ 
  and $\Rmeas_k  \subset \mathbb{R}^n$, respectively.
  

\subsection{Offline Learning Phase}
\label{sec:training}

The objective of this phase is to compute 
  a function $f: \mathbb{R}^n \times \mathbb{R}^m \to \mathbb{R}^n$, 
  such that $\Rpredict_{k+1}=f(\Rmeas_k,\mathcal{U}_k)$, i.e., 
  $f$ returns  $\Rpredict_{k+1}$ 
  given a known input zonotope $\mathcal{U}_k$ and the 
    measurement updated set $\Rmeas_k$ at time-step $k$ 
    such that we can guarantee $x(k+1) \in \Rpredict_{k+1}$ for all $k$.
During this phase, 
  we assume that we have offline an access to an input sequence $u(k)$ and noisy output $z^i(k)$ such that 
\begin{align}
  z^i(k)  &= C^i x(k) + \gamma^i(k),
  \label{eq:training_measurements}
\end{align}
  where the noise $\gamma^i(k)$ 
  is bounded by the zonotope $\mathcal{Z}_{\gamma,i} = \zono{c_{\gamma,i},G_{\gamma,i}}$, 
  i.e.,
  $\gamma^i(k) \in \mathcal{Z}_{\gamma,i}, \forall k$. We have for all sensors vertically combined noisy output $z(k) = \begin{bmatrix} z^{1^T}(k)& ...& z^{q^T}(k) \end{bmatrix}^T$ and similarly for $\gamma$ and $C$. For the sake of clarity, we differentiate the notation of the offline noisy output $z^i(k)$ from the online noisy output $y^i(k)$ and similarly for the measurement noise. 
Given an experiment yielding a sequence of noisy data of length $T$, 
  we can construct the following sequences 
\begin{align}
  \begin{split}
  Z^{+} &= \begin{bmatrix}z(1)&\dots&z(T)\end{bmatrix}, \\
  Z^{-} &= \begin{bmatrix}z(0)&\dots&z(T-1)\end{bmatrix}, \\
  U^{-} &= \begin{bmatrix}u(0)&\dots&u(T-1)\end{bmatrix}.
  \end{split}
  \label{eq:data_sequences}
\end{align}
We further construct
\begin{align*}
  Z&= \begin{bmatrix}z(0)&\dots&z(T)\end{bmatrix}, 
\end{align*}  
and similarly for other signals. The data $D = \begin{bmatrix}
  U^{-} & Z
\end{bmatrix}$ can be from one sensor or multiple sensors. Furthermore, we denote the sequence of 
  \textit{unknown} process noise $w(k)$ as 
    ${W}^{-} =
    \begin{bmatrix}  
        {w}(0) & \dots & {w}(T{-}1) 
    \end{bmatrix}$.
Here,   
  ${W}^{-} \in \mathcal{M}_w$ 
  where $\mathcal{M}_w = \zono{C_{\mathcal{M},w}, G^{(1: \xi T)}_{\mathcal{M},w}}$ 
  is the matrix zonotope resulting from the concatenation of multiple
  noise zonotopes 
  $\mathcal{Z}_w=\zono{c_{w},[g_{w}^{(1)}, \dots, g_{w}^{(\xi)}]}$ 
  as
\begin{equation*}
  \begin{split}
    C_{\mathcal{M},w} &= \begin{bmatrix}c_{w} & \dots & c_{w}\end{bmatrix}, \\
    G^{(1+(i-1)T)}_{\mathcal{M},w} &= \begin{bmatrix} g_{w}^{(i)} & 0_{n \times  (T-1)}\end{bmatrix}, \\
    G^{(j+(i-1)T)}_{\mathcal{M},w} &= \begin{bmatrix} 0_{n \times  (j-1)} &g_{w}^{(i)}  & 0_{n \times  (T-j)}\end{bmatrix}, \\
    G^{(T+(i-1)T)}_{\mathcal{M},w} &= \begin{bmatrix} 0_{n \times (T-1)}& g_{w}^{(i)}\end{bmatrix},
  \end{split}
\end{equation*}
for all $i =\{1, \dots, \xi\}$, $j=\{2,\dots,T-1\}$ \cite{conf:datadriven_reach}. In a similar fashion, we describe the unknown noise and matrix zonotope of 
  $\gamma(k)$ as $\Gamma^{+},\Gamma^{-} \in \mathcal{M}_{\gamma} = \zono{ C_{\mathcal{M},{\gamma}}, G^{(1: \xi T)}_{\mathcal{M},{\gamma}} }$. We denote all system matrices $\begin{bmatrix} A & B \end{bmatrix}$ that are consistent with the data:
\begin{align*}
    \mathcal{N}_{\Sigma} = \{ & \begin{bmatrix} A & B \end{bmatrix} | \; X^{+} = A X^{-} + B U^{-} + W^{-},\\
    & Z^{-}=CX^{-} + \Gamma^{-}, W^{-} \in \mathcal{M}_w, \Gamma^{+}\in \mathcal{M}_\gamma, \\
    & \Gamma^{-} \in \mathcal{M}_\gamma \}.
\end{align*}
By definition, $\begin{bmatrix} A_{\text{tr}} & B_{\text{tr}} \end{bmatrix} \in \mathcal{N}_{\Sigma}$ as $\begin{bmatrix} A_{\text{tr}} & B_{\text{tr}} \end{bmatrix}$ is one of the systems that are consistent with the data. The following theorem finds a set of models $\mathcal{M}_\Sigma$ that over-approximates $\mathcal{N}_{\Sigma}$, i.e., $\mathcal{N}_{\Sigma} \subseteq \mathcal{M}_\Sigma$, which defines $f(\cdot)$ introduced above. For this,
  we aim to determine the mapping of the observation $Z^{+}$ and $Z^{-}$
  to the corresponding state-space region.
Specifically, 
  we construct a zonotope $\mathcal{Z}_{x|z^i(k)} \subset \mathbb{R}^n$ 
  that contains all \textit{possible} $x \in \mathbb{R}^n$
  given $z^i(k)$, $C^i$ and 
  bounded noise $\gamma^i(k) \in \mathcal{Z}_{\gamma,i}$ 
  satisfying \eqnref{eq:training_measurements}, for each $i$.
This can be written as
\begin{align} 
    \mathcal{Z}_{x|z^i(k)} = \Big\{ x \in \mathbb{R}^n \; \Big| \; 
        C^i x = z^i(k) - \mathcal{Z}_{\gamma,i}
    \Big\}.
    \label{eq:measurement_set}
\end{align}
Extending \eqref{eq:measurement_set} to a matrix zonotope allows to find the mapping of $Z^{+}$ and $Z^{-}$ to the state space which is utilized to compute the $\mathcal{M}_\Sigma$. 
  We omit the time index $k$ and sensor index $i$ when possible for simplicity. We assume a prior known upper bound $M$ on the state trajectory, i.e., $M \geq \norm{x}_2$. 
  \begin{lemma}
 \label{lm:sigmaM}
Given input-output trajectories $D = \begin{bmatrix}
  U^{-} & Z
\end{bmatrix}$ of the system \eqref{eq:sys}. Then, the matrix zonotope 
\begin{align}
    \mathcal{M}_\Sigma = (\mathcal{M}^{+}_{x|z} - \mathcal{M}_w) \begin{bmatrix} \mathcal{M}^{-}_{x|z} \\ U^{-}\end{bmatrix}^\dagger
    \label{eq:zonoAB}
\end{align} 
 contains all matrices $\begin{bmatrix}A & B \end{bmatrix}$ that are consistent with the data $D$ and the noise bounds, i.e., $\mathcal{N}_{\Sigma} \subseteq \mathcal{M}_\Sigma$, with $\mathcal{M}^{+}_{x|z} = \zono{C^{+}_{\mathcal{M},x|z},G_{\mathcal{M},x|z}^{(1:\xi T +1)}}$ and  $\mathcal{M}^{-}_{x|z} = \zono{C^{-}_{\mathcal{M},x|z},G_{\mathcal{M},x|z}^{(1:\xi T +1)}}$  where
    \begin{align}
        C^{+}_{\mathcal{M},x|z}  &= V_1 \Sigma_{r \times r}^{-1}P_1^\top\big( Z^{+} - C_{\mathcal{M},{\gamma}} \big),\label{eq:Cplusxgy} \\
         C^{-}_{\mathcal{M},x|z}  &= V_1 \Sigma_{r \times r}^{-1}P_1^\top\big( Z^{-} - C_{\mathcal{M},{\gamma}} \big),\label{eq:Cxgy}\\
        G_{\mathcal{M},x|z}^{(i)} &=  V_1 \Sigma_{r \times r}^{-1}P_1^\top G^{(i)}_{\mathcal{M},{\gamma}}, \quad i=\{1,\dots,\xi T\},\label{eq:Gxgy1} \\
        G_{\mathcal{M},x|z}^{(\xi T +1)} &=M V_2 1_{(n -r) \times T}, \label{eq:Gxgy2} 
    \end{align} 
    for all $M \geq \norm{x}_2$, with $P_1$, $V_1$, $\Sigma$ and $V_2$ 
      obtained from the SVD of $C$.
    Assuming $C$ has rank $r$, then 
    \begin{align}
      C = \begin{bmatrix}
          P_1 & P_2
        \end{bmatrix}
        \begin{bmatrix}
          \Sigma_{r \times r} & 0_{r \times (n-r)} \\ 
          0_{(p-r)\times r} & 0_{(p-r)\times(n-r)}
        \end{bmatrix}
        \begin{bmatrix}
          V_1^\top \\ V_2^\top
        \end{bmatrix},
        \label{eq:svd_of_C}
    \end{align}
    where a matrix with non-positive index is an empty matrix.
    \label{prop:measurement_zonotope}
  \end{lemma}
  \begin{proof}
    From \eqnref{eq:svd_of_C}, 
      we rewrite \eqnref{eq:training_measurements} 
      as ${P_1\Sigma V_1^\top x = z - \gamma}$,
      so $x = V_1\Sigma^{-1}P_1^\top (z - \gamma)$.
    Since $\gamma$ is bounded by ${\mathcal{Z}_{\gamma}= \zono{c_{\gamma},G_{\gamma}}}$,
    we can write  
      \begin{align*}
        x= \underbrace{V_1 \Sigma^{-1}P_1^\top\big( z - c_{\gamma}\big)}_{c_{x|z}} - 
          \underbrace{V_1 \Sigma^{-1}P_1^\top G_{\gamma}}_{G_{x|z}'} \beta, \;\; |\beta| \leq 1.
      \end{align*}
    This set corresponds to all possible $x$ values 
      within the range space of $C$ satisfying \eqnref{eq:training_measurements}.
    By definition, if $r = n$, then ${V_2 = \emptyset}$, 
      $V_1$ spans the domain of $x$,
      and $\zono{c_{x|z},G_{x|z}'}$ sufficiently defines all possible $x$
      satisfying \eqnref{eq:training_measurements}.
    However, 
      if $r < n$, $V_1$ only spans a subset of the domain of $x$.
    To ensure $\mathcal{Z}_{x|z}$ contains all possible $x$
      we include a basis for $\texttt{ker}(C)$ in $G_{x|z}$
    by appending the generator $V_2M$ to $G_{x|z}$, and ensuring $M \geq \|x\|_2$ such that $V_2M$ includes all $x$ values in the directions of $V_2$.
    In both cases for $r$, the generator matrix can be written as
    \begin{align*}
      G_{x|z} 
      = \begin{bmatrix} 
        G_{x|z}' & V_2M 
      \end{bmatrix} 
      = \begin{bmatrix} 
        V_1 \Sigma^{-1}P_1^\top G_{\gamma} & V_2M 
      \end{bmatrix}, 
    \end{align*}
    and the set $\mathcal{Z}_{x|z} = \zono{c_{x|z}, G_{x|z}}$.
    This result extends to the case when $r < p$ using similar argumentation in the respective cases $r = n$ and $r < n$. Considering the matrix version of $\mathcal{Z}_{x|z}$ results in proving $\mathcal{M}^{+}_{x|z}$ and $\mathcal{M}^{-}_{x|z}$. Then, we extend the proof of \cite[Lem.1]{conf:ourjournal} for input-output data: For any $\begin{bmatrix}A & B \end{bmatrix} \in \mathcal{N}_\Sigma$, 
we know that there exists a $W^{-} \in \mathcal{M}_w$ such that 
\begin{align}
    A X^{-} + B U^{-} = X^{+} - W^{-}.
    \label{eq:pf1_1}
\end{align}
Every $W^{-} \in \mathcal{M}_w$ can be represented by a specific choice $\hat{\beta}^{(i)}_{\mathcal{M},w}$, $-1 \leq \hat{\beta}^{(i)}_{\mathcal{M},w} \leq 1$, $i=1,\dots,\xi_{\mathcal{M},w}$, that results in a matrix inside the matrix zonotope $\mathcal{M}_w$:
\begin{align*}
    W^{-} &= C_{\mathcal{M},w} + \sum_{i=1}^{\xi_{\mathcal{M},w}} \hat{\beta}^{(i)}_{\mathcal{M},w} G_{\mathcal{M},w}^{(i)}.
\end{align*}
Rearranging \eqref{eq:pf1_1} and considering $\mathcal{M}^{+}_{x|z}$ and $\mathcal{M}^{-}_{x|z}$ as an over-approximation of $X^{+}$ and $X^{-}$, respectively, yields
\begin{align}
   \begin{bmatrix} A\!\! &\! B \end{bmatrix} {=}\!\! \left(\!\! \mathcal{M}^{+}_{x|z} {-}  C_{\mathcal{M},w} {-}\sum_{i=1}^{\xi_{\mathcal{M},w}} \hat{\beta}^{(i)}_{\mathcal{M},w} G_{\mathcal{M},w}^{(i)} \right)\!\! \begin{bmatrix}
    \mathcal{M}^{-}_{x|z} \\ U^{-} 
   \end{bmatrix}^\dagger
   \label{eq:pf2}
\end{align}
Hence, for all $\begin{bmatrix}A & B \end{bmatrix} \in \mathcal{N}_{\Sigma}$, 
there exists $\hat{\beta}^{(i)}_{\mathcal{M},w}$, ${-1\leq\hat{\beta}^{(i)}_{\mathcal{M},w}\leq 1}$, $i=1,\dots,\xi_{\mathcal{M},w}$, such that \eqref{eq:pf2} holds. Therefore, for all $\begin{bmatrix}A & B \end{bmatrix} \in \mathcal{N}_{\Sigma}$, it also holds that $\begin{bmatrix} A & B \end{bmatrix} \in \mathcal{M}_\Sigma$ as defined in \eqref{eq:zonoAB}, which concludes the proof.
  \end{proof}
Given that we have found a matrix zonotope $\mathcal{M}_\Sigma$ that contains the true system dynamics $\begin{bmatrix} A_{\text{tr}} & B_{\text{tr}} \end{bmatrix} {\in} \mathcal{M}_{\Sigma}$, we can utilize it in computing the time update reachable set $\Rpredict_k$ in the following theorem. 
\begin{theorem}
\label{th:reach_lin}
The set $\Rpredict_k$ over-approximates the exact reachable set, i.e., $\Rpredict \supseteq \mathcal{R}_{k}$ where
\begin{align} \label{eq:Rp}
    \Rpredict_{k+1} = \mathcal{M}_\Sigma ( \Rpredict_{k} \times \mathcal{U}_k ) + \mathcal{Z}_w, 
\end{align}
and $\Rpredict_{0} = \mathcal{X}_0$.
\end{theorem}
\begin{proof}
As $\begin{bmatrix} A_{\text{tr}} & B_{\text{tr}} \end{bmatrix} {\in} \mathcal{M}_{\Sigma}$ according to Lemma~\ref{lm:sigmaM} and starting from the same initial set $\mathcal{X}_0$, it follows that ${\Rpredict_{k}{\supseteq}\mathcal{R}_{k}}$.
\end{proof}

\subsection{Online Estimation Phase using Zonotopes}
\label{sec:estimation_zonotopes}
 
In this subsection, we present the \textit{online estimation phase}.
We are now considering the system \eqnref{eq:underlying_system} 
  with observations \eqnref{eq:observations}.
This phase consists of a time update 
  and a measurement update.
In \secref{sec:training}, we derived the function $f(\cdot)$ 
  for the time update.  
We next present 
  two approaches to perform the measurement update.

\subsubsection{Approach 1 - Reverse-Mapping} 
For this approach,
  we aim to determine the mapping of an observation $y^i(k)$ 
  to the corresponding state-space region.
Similar to Lemma \ref{lm:sigmaM}, 
  we construct a zonotope $\mathcal{Z}_{x|y^i(k)} \subset \mathbb{R}^n$ 
  that contains all \textit{possible} $x \in \mathbb{R}^n$
  given $y^i(k)$, $C^i$ and 
  bounded noise $v^i(k) \in \mathcal{Z}_{v,i}$ 
  satisfying \eqnref{eq:observations}, for each $i$. 
  \begin{proposition}
    Assume $\|x\|_2 \leq K$.
    Given a measurement $y^i(k)$ with noise  
      $v^i(k) \in \mathcal{Z}_{v,i} = \zono{c_{v,i},G_{v,i}}$ 
      satisfying \eqnref{eq:observations},
    the possible states $x$ that correspond to this measurement        
    are contained within the zonotope
    $ \mathcal{Z}_{x|y^i} = \zono{c_{x|y^i},G_{x|y^i}},$
    where
    \begin{equation}
      \begin{split}
        c_{x|y^i} &= V_1 \Sigma_{r^i \times r^i}^{-1}P_1^\top\big( y^i(k) - c_{v,i} \big), \\
        G_{x|y^i} &= \begin{bmatrix} 
          V_1 \Sigma_{r^i \times r^i}^{-1}P_1^\top G_{v,i} & V_2 M
      \end{bmatrix},
      \end{split}
      \label{eq:prop_1_eqn}
    \end{equation} 
    for all $M \geq K$, with $P_1$, $V_1$, $\Sigma$ and $V_2$ 
      obtained from the SVD of $C^i$ as in \eqref{eq:svd_of_C}.
    \label{prop:measurement_zonotope}
  \end{proposition}
  \begin{proof}
  The proof follows immediately from Lemma \ref{lm:sigmaM}.
    \end{proof}

  \begin{remark}
    In our case, $\mathcal{Z}_{x|y^i(k)}$ will eventually be 
      intersected with $\Rpredict_k = \zono{\tilde{c}_k, \tilde{G}_k}$.
    It is therefore sufficient to set $M \geq \texttt{radius}(\Rpredict_k) + \|V_2^\top \tilde{c}_k\|_2$
      instead of the more conservative $M \geq \norm{x}_2$,
    where 
      $\texttt{radius}(\Rpredict_k)$ 
      returns the radius of a minimal hyper-sphere containing $\Rpredict_k$
      \cite{conf:cora}.
  \end{remark}
Having determined the sets $\mathcal{Z}_{x|y^i(k)}$ 
  for all $i \in \{1,\dots,q \}$, we can 
  compute the measurement updated set $\Rmeas_k$ given 
  the predicted set $\Rpredict_k$ and each measurement set $\mathcal{Z}_{x|y^i(k)}$
  as   
\begin{align}
  \Rmeas_{k} = \Rpredict_k \cap_{i=1}^q \mathcal{Z}_{x|y^i(k)},
  \label{eq:svd_intersection}
\end{align}  
which can be performed using the standard intersection operations 
  presented in \cite{conf:cora,conf:set-diff}.

\subsubsection{Approach 2 - Implicit Intersection}

Contrary to Approach 1, 
  here, we do not \textit{explicitly} 
  determine the sets $\mathcal{Z}_{x|y^i(k)}$.
Instead, $\hat{\mathcal{R}}_k$ is determined directly from the set $\tilde{\mathcal{R}}_k$, the measurements $y^i(k)$
  and some weights $\lambda_k^i$ for $i \in \{1,\dots,q\}$.
We then optimize over the weights
  to minimize the volume of $\hat{\mathcal{R}}_k$.

\begin{proposition}
  The intersection of $\Rpredict_{k}  = \zono{ \tilde{c}_{k}, \tilde{G}_{k}} $ 
    and the $q$ regions for $x$ corresponding to $y^i(k)$ with noise  
      $v^i(k) \in \mathcal{Z}_{v,i} = \zono{c_{v,i},G_{v,i}}$ 
      satisfying \eqnref{eq:observations} 
    can be over-approximated by the zonotope 
    $\Rmeas_{k} = \zono{ \hat{c}_{k},\hat{G}_{k} } $ with
  \begin{align}
    \hat{c}_{k} &=  
        \tilde{c}_{k} + 
        \sum\limits_{i = 1}^q 
        \lambda_{k}^{i}\Big(y^{i}(k) - C^i \tilde{c}_{k} - c_{v,i} \Big), \label{eq:C_lambda}\\
    \hat{G}_{k} &= 
      \begin{bmatrix} (I - \sum\limits_{i = 1}^q \lambda_{k}^i C^i ) 
      \tilde{G}_{k} & -\lambda_{k}^{1} G_{v,1} & \dots &
      - \lambda_{k}^{q} G_{v,q} \end{bmatrix},
    \label{eq:G_lambda}
  \end{align}
  where 
    $\lambda_{k}^{i} \in \R^{n \times p_i}$ for $i \in \{1,\dots,q \}$ 
  are weights.
  \label{prop:zonotope_opt_method}
\end{proposition}
\begin{proof}
The proof is based on \cite[Prop.1]{conf:stripzono} 
  but with zonotopes as measurements instead of strips. 
Let $x \in \Rpredict_{k} \cap \mathcal{Z}_{x|y^1} \cap \dots \cap \mathcal{Z}_{x|y^q}$. 
Then there exists a $z$ such that $x = \tilde{c}_{k} + \tilde{G}_{k} z$.
   Adding and subtracting $\sum_{i = 1}^q \lambda_{k}^{i} C^{i}\tilde{G}_k z$ yields
   \begin{equation}
      x = \tilde{c}_k +\sum\limits_{i = 1}^q \lambda_{k}^{i} C^{i} \tilde{G}_k z 
      + 
      ( I - \sum\limits_{i = 1}^q \lambda_{k}^{i} C^{i}) \tilde{G}_k z. 
      \label{equ:x_zono_2}
   \end{equation}
  From \eqnref{eq:observations}, we obtain 
    $C^i x  =  y^i - c_{v,i} - G_{v,i} d^i. $
 Using $x = \tilde{c}_{k} + \tilde{G}_{k} z$ yields
    $ C^{i} \tilde{G}_k z  = y^i(k) - C^{i} \tilde{c}_k  - c_{v,i} - G_{v,i} d^i$, 
  which we insert into \eqref{equ:x_zono_2} to obtain
  \begin{equation*}
    \begin{aligned}
      x &= \tilde{c}_k +\sum\limits_{i = 1}^q \lambda_{k}^{i} \Big(y^i(k)  - C^{i} \tilde{c}_k  - c_{v,i} - G_{v,i} d^i \Big) \\
        & \;\;\;   + \Big( I - \sum\limits_{i = 1}^{q} \lambda_{k}^{i} C^i \Big) \tilde{G}_k z, \\
        &= \underbrace{
          \begin{bmatrix} 
          (I - \sum\limits_{i = 1}^{q} \lambda_{k}^{i} C^i ) \tilde{G}_k & - \lambda_{k}^{1} G_{v,1} & \dots & 
          - \lambda_{k}^{q} G_{v,q} \end{bmatrix}
          }_{\hat{G}_{k}}  
          \!\!  
          \underbrace{
            \begin{bmatrix} 
              z \\ d^1 \\ \vdots \\ d^{q} 
          \end{bmatrix}}_{z^b}  \\
        & \;\;\; +  \underbrace{\tilde{c}_k+\sum\limits_{i = 1}^{q} \lambda_{k}^{i} (y^i(k)  - C^i \tilde{c}_k -c_{v,i})}_{\hat{c}_{k}}
          = \hat{G}_{k} z^b + \hat{c}_{k}.
    \end{aligned}
  \end{equation*}
  Note that $z^b \in [-1,1]$ since $d^i \in  [-1,1]$ and $z \in [-1,1]$.
  $\hat{R}_k$ adheres to \defref{def:zonotope} with center $\hat{c}_{k}$ and generators $\hat{G}_{k}$.
\end{proof}

As in \cite{conf:set-diff},
we find the optimal weights $\lambda_{k}^{i} \in \R^{n \times p_i}$ from 
\begin{align}
  \bar{\lambda}^*_k = \argmin_{\bar{\lambda}_k} \norm{\hat{G}_{k}}^2_F,
  \label{eq:opt_problem}
\end{align}
where $\bar{\lambda}_k = [ \lambda_{k}^{1} \dots \lambda_{k}^{q}]$. 

The online estimation phase 
  is illustrated in the block diagram of \figref{fig:method}.
The detailed estimation phase
  is presented in \algref{alg:svd_method}.
The function \textit{measZon()} 
  executes \propref{prop:measurement_zonotope},
  and \textit{optZon()} 
  \propref{prop:zonotope_opt_method}.
The function \textit{reduce}$(\Rpredict_{k+1})$ reduces the order of $\Rpredict_{k+1}$
  using the method proposed in \cite{Girard2005}, 
    which ensures the number of generators in $\Rpredict_{k+1}$ 
    remains relatively low, avoiding potential 
    tractability issues after multiple iterations.

\begin{algorithm}
  ${\Rmeas}_0 = \mathcal{X}_0$ \\
  $k = 1$ \\
  \While{True}{
  $\Rpredict_{k} = f(\Rmeas_{k-1},\zono{u(k-1),0})$ using \eqnref{eq:Rp}  \\
  \uIf{Approach 1}{
  \ForEach{$i \in \{1,\dots,q \}$}{
    $\mathcal{Z}_{x|y^i(k)} = \textit{measZon}\big(y^i(k),\mathcal{Z}_{v,i},C^i\big)$ 
    using \eqnref{eq:prop_1_eqn} 
  }
  $\Rmeas_{k} = \Rpredict_k \bigcap_{i=1}^q \mathcal{Z}_{x|y^i(k)}$ \\
  }
  \uIf{Approach 2}{
  $\zono{\hat{c}_{k},\hat{G}_{k}} = \textit{optZon}(\Rpredict_k,y(k),C,\mathcal{Z}_{v})$ \\
  $\hat{G}_{k}^*, \; \bar{\lambda}^* \leftarrow $ Solve \eqnref{eq:opt_problem} \\
  $\Rmeas_{k} = \zono{\hat{c}_{k}, \hat{G}_{k}^* }$
  }
  $\Rpredict_{k} = \textit{reduce}(\Rmeas_{k})$ using \cite{Girard2005} \\
  $k \leftarrow k + 1$ 
  }
  \caption{\textit{Online Estimation Phase}}
  \label{alg:svd_method}
\end{algorithm}

\subsection{Online Estimation Phase using Constrained Zonotopes}
\label{sec:estimation_conzonotopes}

When intersecting 
  zonotopes, the result is an over-approximation of the true intersection.
However, it is possible to 
  determine the \textit{exact} intersection of constrained zonotopes.


\begin{definition}(Constrained zonotope \cite{conf:const_zono})
    An $n$-dimensional constrained zonotope is
    \begin{equation}
       \hspace{-2mm} \mathcal{C} = 
         \setdef[x\in\mathbb{R}^n]{x=c_{\mathcal{C}}+G_{\mathcal{C}} \beta, \ A_{\mathcal{C}} \beta=b_{\mathcal{C}}, \, \norm{\beta}_\infty\leq 1}, 
        \label{eq:conszono}
    \end{equation}
    where $c_{\mathcal{C}} \in \R^n$ is the center, 
    $G_{\mathcal{C}}$ $\in$ $\R^{n \times n_g}$ the generator matrix 
    and $A_{\mathcal{C}} \in $ $\R^{n_c \times n_g}$ and $b_{\mathcal{C}} \in \R^{n_c}$ 
    the constraints. 
    In short, 
      we write $\mathcal{C}= \zono{c_{\mathcal{C}},G_{\mathcal{C}},A_{\mathcal{C}},b_{\mathcal{C}}}$.
    \label{def:con_zonotope}
\end{definition}

When using constrained zonotopes, 
  we replace the time and measurement updated sets
  $\Rpredict_k$ and $\Rmeas_k$ by the constrained zonotopes
  $\Cpredict_k$ and $\Cmeas_k$, respectively.

\subsubsection{Approach 1 - Reverse-Mapping}

This approach works directly with constrained zonotopes.
The sets $\mathcal{Z}_{x|y^i(k)}$ of \propref{prop:measurement_zonotope} 
  are constrained 
  zonotopes with no $A_\mathcal{C},b_\mathcal{C}$ constraints.
The intersection in \eqnref{eq:svd_intersection} becomes 
    ${\Cmeas_k = \Cpredict_k \cap_{i=1}^q \mathcal{Z}_{x|y^i(k)}}$
which can be performed as described in \cite{conf:const_zono}.

\subsubsection{Approach 2 - Implicit Intersection}

We adapt \propref{prop:zonotope_opt_method} 
  to use constrained zonotopes.
\begin{proposition}
  The intersection of $\Cpredict_{k}  = \zono{ \tilde{c}_k,\tilde{G}_k,\tilde{A}_k,\tilde{b}_k } $ 
    and $q$ regions for $x$ corresponding to $y^i(k)$
      as in \eqnref{eq:observations} 
    can be described by the constrained zonotope 
    $\Cmeas_k=\zono{\hat{c}_k,\hat{G}_k,\hat{A}_k,\hat{b}_k}$ 
    with weights $\lambda_k^i \in \mathbb{R}^{n\times p_i}$ for $i \in \{1,\dots,q\}$ where 
  \begin{align}
      \hat{c}_k &=  \tilde{c}_k +\sum\limits_{i = 1}^{q} \lambda_{k}^{i} \big(y^i(k)  - C^{i} \tilde{c}_k -c_{v,i} \big), \nonumber\\
      \hat{G}_k &= 
      \begin{bmatrix} 
        (I - \sum\limits_{i = 1}^{q} \lambda_{k}^{i} C^{i} ) \tilde{G}_k
        & - \lambda_{k}^{1} G_{v,1} 
        & \dots 
        & -\lambda_{k}^{q} G_{v,q} 
      \end{bmatrix}, 
      \label{eq:G_lambdapr3}\\
      \hat{A}_k &= 
      \begin{bmatrix}
        \tilde{A}_k  & 0 & \dots & 0 \\
      C^{1} \tilde{G}_k\!\! &\!\!G_{v,1} &\!\!\dots\!\!&\!\! 0 \\
    \vdots \!\!&\!\!  &\!\!\ddots\!\! &\!\!\\
    C^{q} \tilde{G}_k\!\! &\!\! 0  &\!\! \dots\!\! &\!\! G_{v,q} \end{bmatrix}, \\ 
      \hat{b}_k &= 
      \begin{bmatrix}
        \tilde{b}_k \\
        y^1(k) - C^{1} {c}_k - c_{v,1}\\ 
        \vdots \\ 
        y^q(k) - C^{q} {c}_k - c_{v,q}
      \end{bmatrix}. 
      \label{eq:Abbarconstzono}
  \end{align}
  \vspace{-4mm}
\label{prop:con_zonotopes}
\end{proposition}
\begin{proof}
  We follow a similar approach to \cite[Thm. 6.3]{conf:set-prv} and \cite{conf:const_zono},
    but extend the proof by defining measurement sets as zonotopes instead of strips.  
   $\mathcal{Z}_{x|y^i}$ refers to $\mathcal{Z}_{x|y^i(k)}$ 
     unless specified otherwise.
    Let 
      $x_k \in \tilde{\mathcal{C}}_k \cap \mathcal{Z}_{x|y^1} \cap \dots \cap \mathcal{Z}_{x|y^q} $, 
      then there exists a $z_k\in\left[-1,1\right]$ such that
    \begin{align}
        x_k = \tilde{c}_k + \tilde{G}_k z_k, \hspace{5mm} \tilde{A}_k z_k = \tilde{b}_k. \label{equ:x_zono}  
    \end{align}
    Using \eqnref{eq:observations} and the measurement noise $\zono{ c_{v,i},G_{v,i} }$, we write
      \begin{align}
      C^{i} x  =  y^i(k) - c_{v,i} - G_{v,i} d^i,
      \label{equ:bj}
      \end{align}
    where $d^i \in [-1,1]$. Inserting \eqref{equ:x_zono} into \eqref{equ:bj} yields 
    \begin{align}
      C^{i} \tilde{G}_k z_k = y^i(k) - C^{i} \tilde{c}_k - c_{v,i} - G_{v,i} d^i,  \label{eq:CGz}
    \end{align}
    which, combined 
      with \eqref{equ:x_zono}, yields
    \begin{align}
        \hspace{-1mm} \underbrace{\begin{bmatrix}
          \tilde{A}_k             &  0        & \hspace{-2mm} \dots \hspace{-2mm} & \hspace{-2mm} 0 \\
          C^{1} {G}_k       &  G_{v,1} & \hspace{-2mm} \dots \hspace{-2mm} & \hspace{-2mm} 0 \\
          \vdots            &           & \hspace{-2mm} \ddots \hspace{-2mm} & \hspace{-2mm} \\
          C^{q} {G}_k       &    0      & \hspace{-2mm} \dots \hspace{-2mm}  & \hspace{-2mm} G_{v,q} 
        \end{bmatrix}}_{\hat{A}_k} 
        &
        \underbrace{
          \begin{bmatrix} z_k \\ d^1 \\ \vdots \\ d^{q} \end{bmatrix}
        }_{z_b}
        \hspace{-1mm} = \hspace{-1mm}
        \underbrace{
          \begin{bmatrix}
            \tilde{b}_k \\ y^1(k) - C^{1} {c}_k - c_{v,1} \\ \vdots \\ y^q(k) - C^{q} {c}_k - c_{v,q} 
          \end{bmatrix}
        }_{\hat{b}_k}.
        \label{equ:z_bj}
    \end{align}
    Adding and subtracting $\sum_{i = 1}^q \lambda_{i,k} C^{i} \tilde{G}_k z_k$ to \eqref{equ:x_zono} yields 
    \begin{equation}
      x_k = 
      \tilde{c}_k + \sum_{i = 1}^q \lambda^i_{k} C^{i} \tilde{G}_k z_k 
      + ( I - \sum_{i = 1}^q  \lambda^i_{k} C^{i}) \tilde{G}_k z_k. 
      \label{equ:x_zono2}
    \end{equation}
    If we now insert \eqref{eq:CGz} into \eqref{equ:x_zono2}, 
      we obtain 
    \begin{align*}
      x &= 
      \underbrace{
        \begin{bmatrix} 
          (I - \sum\limits_{i = 1}^{q} \lambda_{k}^{i} C^{i} ) \tilde{G}_{k} & -\lambda_{k}^{1}  G_{v,1} & \dots & -\lambda_{k}^{m_i}  G_{v,q} 
        \end{bmatrix}
      }_{\hat{G}_{k}} z_b \\
      & \;\;\; \; +  
      \underbrace{
        \hat{c}_{k-1} +\sum\limits_{i = 1}^{q} \lambda_{k}^{j} \big( y^i(k)  - C^{i} \tilde{c}_{k} -c_{v,i} \big)}_{\hat{c}_{k}}
      = 
      \hat{G}_{k} z_b + \hat{c}_{k}.
    \end{align*}
    Hence, 
    $x(k) \in \Cmeas_k$ and $(\Cpredict \cap \mathcal{Z}_{x|y^1} \cap \dots \cap \mathcal{Z}_{x|y^q}) \subseteq \Cmeas_k$.
    Conversely, 
      let $x(k) \in \Cmeas_k$.
    Then, there exists a $z_b$ 
      such that \eqref{eq:conszono} in \defref{def:con_zonotope} is satisfied. 
    Partitioning $z_b$ into $z_b =[z_k, d^1 \dots ,d^q]^T$, 
      it follows that we can construct 
      a constrained zonotope 
      $\tilde{\mathcal{C}}_k=\{ \tilde{c}_k,\tilde{G}_k,\tilde{A}_k,\tilde{b}_k \}$ 
      given that $\|z_k\|_\infty \leq 1$. 
    Thus, $x(k) \in \tilde{\mathcal{C}}$. 
    Similarly, we can get the constraints in \eqref{equ:bj}. 
    Inserting \eqref{equ:x_zono} in \eqref{eq:CGz} 
      results in obtaining all the equations in \eqref{equ:bj}. 
    Therefore, $x(k) \in \mathcal{Z}_{x|y^i(k)}$, $\forall i \in \{1,\dots,q\}$. 
    Thus, $x(k) \in (\tilde{\mathcal{C}}_k \cap \mathcal{Z}_{x|y^1} \cap \dots \cap \mathcal{Z}_{x|y^q})$ 
      and $\Cmeas_k  \subseteq (\tilde{\mathcal{C}}_k \cap \mathcal{Z}_{x|y^1} \cap \dots \cap \mathcal{Z}_{x|y^q})$, 
      which concludes the proof. 
 \end{proof}

%% file: Sections/4-evaluation.tex
\section{Evaluation} 
\label{sec:evaluation}

We evaluate our method by considering an input-driven variant 
  of the rotating target described in \cite{conf:set-diff}.
We set
\begin{align}
  A_{\text{tr}} = \begin{bmatrix}
    0.9455 & -0.2426 \\
    0.2486 & 0.9455
  \end{bmatrix},
  \hspace{5mm} 
  B_{\text{tr}} = \begin{bmatrix}
    0.1 \\ 0 
  \end{bmatrix}
\end{align}
with $q=3$ measurements parameterized as follows
\begin{align*}
  &C^1 = \begin{bmatrix}
    1 & 0.4    
  \end{bmatrix}, 
 C^2 = \begin{bmatrix}
    0.9 & -1.2  
  \end{bmatrix},  
  C^3 = \begin{bmatrix}
    -0.8 & 0.2 \\ 0 & 0.7   
  \end{bmatrix}, \\ 
 & \mathcal{Z}_{v,1} = \zono{0,1}, 
   \mathcal{Z}_{v,2} = \zono{0,1}, 
  \mathcal{Z}_{v,3} = \zono{[0\;\; 0]^\top, I_2}. 
\end{align*}
The noise signals are characterized by the zonotopes 
   ${\mathcal{Z}_{\gamma} = \zono{[0\;\; 0]^\top, 0.02I_2}}$ and 
   $\mathcal{Z}_{w} = \zono{[0\;\; 0]^\top, 0.02I_2}$.
We run the \textit{offline learning phase} with $T = 500$
  and inputs sampled uniformly from the set $\mathcal{U} = \zono{0,10}$. 
The noise signals $v^i(k)$, $w(k)$ and $\gamma(k)$ are sampled uniformly
from their respective zonotope sets 
  using the command \textit{randPoint}$(\mathcal{Z})$ 
  as described in \cite{conf:cora}.

After learning $f(\cdot)$, 
  we run the \textit{online estimation phase}.
The initial state set is  
   $\mathcal{X}_0 = \zono{ [0\;\; 0]^\top, 15I_2 }$
   and the true initial state is $x(0) = \begin{bmatrix} -10 & 10 \end{bmatrix}^\top$.
Once again, we sample the inputs uniformly from $\mathcal{U}$.
We evaluate both the zonotope and constrained zonotope 
  methods,
  each time using either of the 
  two proposed measurement update approaches.
\figref{fig:sim_bounds} shows the bounds of $\Rmeas_k$
  in the $x_1$ state dimension
  for both approaches.
\figref{fig:sim_bounds_con} shows the equivalent results when our method
  uses constrained zonotopes.
As expected, 
  $x(k)$ is always contained within $\Rmeas_k$ (or $\Cmeas_k$) 
  at each time step.
Although both measurement update approaches yield similar set sizes on average,
  the set evolution of Approach 2 is comparatively smoother.

Furthermore, 
  we compare our results with \textit{N4SID} subspace identification 
  \cite{VANOVERSCHEE199475} combined with 
  a Kalman filter (KF).
In \figref{fig:sim_snapshot}, we show the sets 
  $\Rmeas_k$ and $\Cmeas_k$,
  using either measurement update approach, using zonotopes or constrained zonotopes.
We also show the ellipse corresponding to the $3\sigma$ uncertainty bound of the 
  KF estimate, indicating that our estimator provides state sets 
  comparable in size to that of the KF. We should mention that KF bounds come without any guarantees. 

\begin{figure}[t!]
  \centering
  \begin{subfigure}[t]{0.49\textwidth} 
      \centering
      \includegraphics[width=0.9\linewidth]{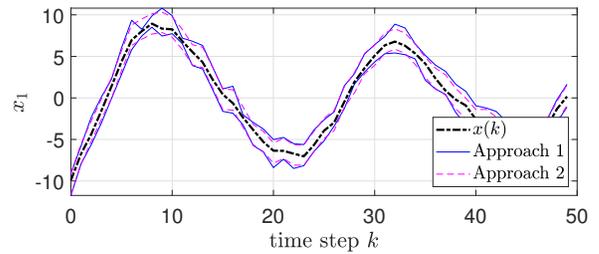}
      \caption{Using zonotopes showing bounds of $\Rmeas_k$ in $x_1$}
      \label{fig:sim_bounds}
      \vspace*{2mm}
  \end{subfigure}
  \begin{subfigure}[t]{0.49\textwidth}
      \centering
      \includegraphics[width=0.9\linewidth]{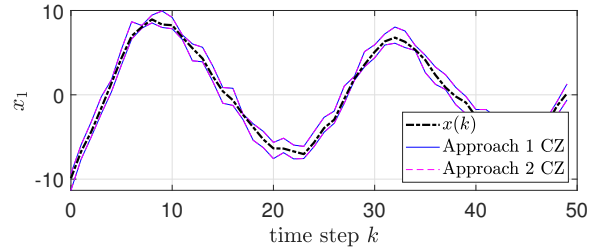}
      \caption{Using constrained zonotopes showing bounds of $\Cmeas_k$ in $x_1$}
      \label{fig:sim_bounds_con}
      \vspace*{2mm}
  \end{subfigure}
  \caption{Bounds of the set $\Rmeas_k$ in (a), and $\Cmeas_k$ in (b),
    projected onto the first state dimension $x_1$ of $x(k)$ 
    using measurement update approaches 1 and 2.}
    \label{fig:sim}
  \vspace*{0mm}
\end{figure}
\begin{figure}[t!]
  \centering 
  \includegraphics[width=0.7\linewidth]{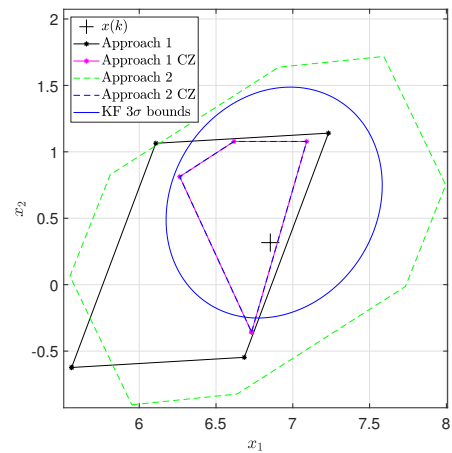}
  \caption{Sets $\Rmeas_k$ using measurement update 
    approaches 1 and 2, 
    and the equivalent sets 
      $\Cmeas_k$ using constrained zonotopes (\textit{CZ}),
      compared to the KF's $3\sigma$ confidence bounds.}
  \vspace*{-4mm}
  \label{fig:sim_snapshot}
\end{figure}

Referring to both \figref{fig:sim} and \figref{fig:sim_snapshot}, it is clear that
  the constrained zonotopes yield smaller state sets at each time step.
However, this comes at the cost of increased computational load.
Running our simulations on a Dell laptop with an 8-core i5-8365U 
  processor at 1.6GHz, the average computation time per iteration for Approach 1 
   increased from $0.656$sec to $1.267$sec.
  when using constrained zonotopes;
for Approach 2, the corresponding times were $0.221$sec and $0.971$sec, respectively.
For all our approaches, 
  we observed that reducing the order of the sets to $5$, 
  which reduces the number of generators in $\Rmeas$ 
  (or $\Cmeas$), was critical to keep the computational load low.

%% file: Sections/5-conclusions.tex
\section{Conclusions and Recommendations} 
\label{sec:conclusions}
In this paper, 
  we introduced a novel zonotope-based method to perform set-based state estimation
  with set containment guarantees 
  using a
  data-driven set propagation function. We presented an approach to compute the set of model that is consistent with the data and noise bounds given input-output data. Then, we presented two approaches to perform the measurement update
  which merges the time updated state set with the observed measurements.
We extended our method to use constrained zonotopes,
  which yielded smaller state sets at the cost of increased computational load.
Our results show state sets comparable in size to 
  the $3\sigma$ uncertainty bounds 
  obtained when running \textit{N4SID} subspace identification and a Kalman filter, 
  but with the added feature of set-containment guarantees
  and without requiring any knowledge 
  of the statistical properties of the noise.

Future work includes evaluating our proposed estimator on 
  real-world examples
  as well as gaining more insight into the limitations of our method
  when applied to more complex dynamical systems.
Additionally, improving the zonotope intersection operation  
  to lessen the degree of over-approximation of the resultant state set  
  would yield tighter state set estimates at each time step.
\section*{Acknowledgement}
This work was supported by the Swedish Research Council, the Knut and Alice Wallenberg Foundation, the Democritus project on Decision-making in Critical Societal Infrastructures by Digital Futures, and the European Unions Horizon 2020 Research and Innovation program under the CONCORDIA cyber security project (GA No. 830927). 